\RequirePackage{amsmath}
\documentclass[runningheads, envcountsame]{llncs}
\usepackage[T1]{fontenc}
\usepackage{graphicx}

\usepackage{amssymb}
\usepackage{hyperref}
\usepackage{cleveref}
\usepackage{ifthen}
\usepackage{tikz}
\usepackage{color}
\usepackage[misc,geometry]{ifsym}

\usetikzlibrary{calc, decorations.pathreplacing, decorations.pathmorphing}
\tikzstyle{vertex}=[circle, fill, inner sep=0pt, minimum size=3pt]
\newcommand{\vertex}{\node[vertex]}
\tikzset{snake it/.style={decorate, decoration=snake}}

\spnewtheorem{observation}[theorem]{Observation}{\bfseries}{\itshape}
\spnewtheorem*{remarknonum}{Remark}{\itshape}{}

\let\given\givenbase

\def\prob{\ensuremath\mathbb{P}}
\def\expect{\ensuremath\mathbb{E}}
\newcommand*\dd{\mathop{}\!\mathrm{d}}

\let\originalleft\left
\let\originalright\right
\renewcommand{\left}{\mathopen{}\mathclose\bgroup\originalleft}
\renewcommand{\right}{\aftergroup\egroup\originalright}

\begin{document}
\title{Approximation Ineffectiveness of a Tour-Untangling Heuristic\thanks{Supported by NWO grant OCENW.KLEIN.176.}} 
%
%
\author{Bodo Manthey \and Jesse van Rhijn \Letter}

\authorrunning{B. Manthey \and J. van Rhijn}

\institute{University of Twente, Enschede, The Netherlands
\email{\{b.manthey,j.vanrhijn\}@utwente.nl}}
\maketitle              
\begin{abstract}
    We analyze a tour-uncrossing heuristic for the Euclidean Travelling Salesperson Problem,
    showing that its worst-case approximation
    ratio is $\Omega(n)$ and its average-case approximation ratio is $\Omega(\sqrt{n})$
    in expectation.
    We furthermore
    evaluate the approximation performance
    of this heuristic numerically on average-case instances,
    and find that it performs far better than the average-case lower bound suggests.
    This indicates a shortcoming in the approach we use for our analysis, which
    is a rather common method in the analysis of local search heuristics.
\keywords{Travelling salesperson problem \and Local search \and Probabilistic analysis.}
\end{abstract}

\section{Introduction}

The Travelling Salesperson Problem (TSP)
is a classic example of an NP-hard combinatorial
optimization problem \cite{korteCombinatorialOptimizationTheory2000}.
Different variants of the problem exist, with one of the most
studied variants being the Euclidean TSP. In this version, the weight of an edge
is given by the Euclidean distance between its endpoints.
Even this restricted version is NP-hard \cite{papadimitriouEuclideanTravellingSalesman1977}.

Due to this hardness, practitioners often turn to approximation algorithms and heuristics for
the TSP. One simple heuristic is 2-opt \cite{aartsLocalSearchCombinatorial2003}. In each iteration of this heuristic,
one searches for a pair of edges in the tour that can be replaced by a different pair, such
that the total length of the tour decreases. Although this
heuristic performs quite well in practice \cite[Chapter 8]{aartsLocalSearchCombinatorial2003},
it may require an exponential number of iterations to converge even in the plane
\cite{englertSmoothedAnalysis2Opt2016}.

Interestingly, 
Van Leeuwen \& Schoone showed that a restricted variant of 2-opt
in which one only removes intersecting edges terminates in $O(n^3)$ iterations
in the worst case \cite{vanleeuwenUntanglingTravelingSalesman1980}.
For convenience, we refer to this variant as X-opt.
More recently, da Fonseca et al.\ \cite{dafonsecaLongestFlipSequence2022} analyzed this heuristic once more,
extending the results to matching problems and showing a bound
of $O(tn^2)$ for instances where all but $t$ points are in convex position.
Their work builds on previous related work
on computing uncrossing matchings \cite{biniazFlipDistancePlane2019}.

The insight that removing intersecting edges improves the tour is a key
intuition behind 2-opt. However, not all 2-opt iterations remove intersections.
Indeed, 2-opt has proved extremely effective also for non-metric TSP instances,
where there is no notion of intersecting edges at all.

This raises the question of approximation performance: can one
get away with using X-opt instead of 2-opt at minimal cost to the approximation
guarantee, thereby ensuring an efficient heuristic for TSP instances in the plane?
The approximation ratio of 2-opt has long been known to sit
between $\Omega\left(\frac{\log n}{\log\log n}\right)$ and $O(\log n)$ for
$d$-dimensional Euclidean instances \cite{chandraNewResultsOld1999},
and has recently been settled to $\Theta\left(\frac{\log n}{\log \log n}\right)$
for 2-dimensional instances \cite{brodowskyApproximationRatioKOpt2021}.
However, no previous work seems to have discussed X-opt.

We analyze this simpler case here, showing an approximation ratio of $\Omega(n)$ in the worst case and
$\Omega(\sqrt{n})$ in the average case. This answers our previously raised question in the negative; in order
to obtain a good approximation ratio, one must allow for iterations that improve the tour
without removing intersections. 
Especially the average-case result stands in stark contrast to the average-case approximation ratio
of 2-opt, which is known to be $O(1)$ \cite{chandraNewResultsOld1999}. 

We also perform a numerical experiment, which presents a different picture
from our formal results. To within the precision we
are able to achieve, our experiments indicate an average-case approximation ratio of X-opt
of $O(1)$. We consider this evidence that the techniques we use
to obtain the average-case bound of
$\Omega(\sqrt{n})$, which are standard techniques used to perform probabilistic analyses of local search
heuristics, fall short of explaining the true practical performance
of X-opt.

\subsection{Definitions and Notation}

Given two points $x, y \in \mathbb{R}^2$, we define $d(x, y)$
as the Euclidean distance between $x$ and $y$. We define
$L(x, y)$ as the line segment between $x$ and $y$.
By an abuse of notation, if $e = \{x, y\}$ with $x \neq y$, we write
$L(e) = L(x, y)$. We write $\ell(e) = d(x, y)$ for the
length of $L(e)$.

Let $X \subseteq \mathbb{R}^2$. For a set $E \subseteq \{e \in 2^X \given |e| = 2\}$,
we write $\ell(E) = \sum_{e \in E}\ell(e)$. If
$L(e) \cap L(f) \neq \emptyset$ for some $e, f \in E$, i.e., some of the line segments
represented by the edges
in $E$ intersect, then we say $E$ is crossing. Conversely,
if $E$ is not crossing, then it is noncrossing.
In particular, a local optimum for X-opt is exactly a
noncrossing tour.

Given $A \subset \mathbb{R}^2$ and $x \in \mathbb{R}^2$, we define
the distance between $A$ and $x$ as $d(x, A) = \min_{y \in A} d(x, y)$.
Note that $d(x, A)$ might not exist, for instance, if $A$
is an open set. However, we will only consider sets for which $d(x, A)$ is
well-defined.

Let $A$ be a rectangular region in $\mathbb{R}^2$. Let $X$ be a set of $n$
points in $A$. We call the four line segments that make up the boundary
of $A$ the edges of $A$.
For an edge $e$ of $A$, let $x_e = \arg\min_{x \in X} d(x, e)$
be the point closest to $e$.
We call $X$ nice for $A$
if for each pair of edges $e$, $f$ of $A$, it holds that
$x_e \neq x_f$.

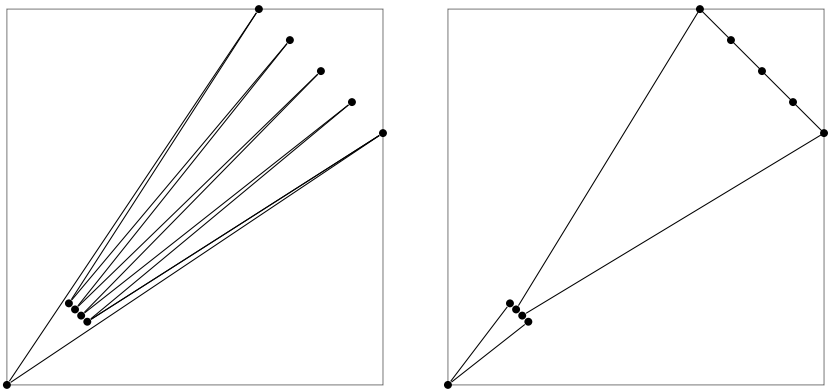
\begin{figure}
    \centering
    \begin{tikzpicture}[scale=5]
        \pgfmathsetmacro\m{5};
        \pgfmathsetmacro\eps{0.33};
        \pgfmathsetmacro\eta{\eps/5};
        \draw[gray] (0,0) -- (1,0) -- (1,1) -- (0,1) -- (0,0);
        \vertex[] (x) at (0, 0) {};

        \pgfmathsetmacro\width{1.414*\eps*\eps / (2 - \eps)};
        \pgfmathsetmacro\offset{(\eps*1.414 - \width) / 2};
        
        \foreach \i in {1,...,\m} {
        
            \pgfmathsetmacro\delta{\eps/(\m-1)};
            \pgfmathsetmacro\kappa{\width/(\m-1)};
            
            \vertex[] (a) at ({1-\eps + (\i-1)*\delta},{1 - (\i-1)*\delta}) {};

            \ifthenelse{\i>1}{\draw (a) -- (b);}{}
            
            \ifthenelse{\i=1}{\draw (x) -- (a);}{}
            \ifthenelse{\i=\m}{\draw (x) -- (a);}{}
            
            \ifthenelse{\i<\m}{\vertex[] (b) at ({((\i+1)*\kappa + \offset)/1.414},{(\eps - \i*\kappa)/1.414}) {};}{}
            
            \draw (a) -- (b);
        }
    \end{tikzpicture}\hspace{2em}
    \begin{tikzpicture}[scale=5]
        \pgfmathsetmacro\m{5};
        \pgfmathsetmacro\eps{0.33};
        \pgfmathsetmacro\eta{\eps/5};
        
        \draw[gray] (0,0) -- (1,0) -- (1,1) -- (0,1) -- (0,0);
        \vertex[] (x) at (0, 0) {};

        \pgfmathsetmacro\width{1.414*\eps*\eps / (2 - \eps)};
        \pgfmathsetmacro\offset{(\eps*1.414 - \width) / 2};
        \pgfmathsetmacro\middle{(\m-1)/2};
        \pgfmathsetmacro\middlee{(\m-1)/2+1};
        
        \foreach \i in {1,...,\m} {
            \pgfmathsetmacro\delta{\eps/(\m-1)};
            \pgfmathsetmacro\kappa{\width/(\m-1)};
            
            \vertex[] (a) at ({1-\eps + (\i-1)*\delta},{1 - (\i-1)*\delta}) {};
            
            \ifthenelse{\i<\m}{\vertex[] (b) at ({((\i+1)*\kappa + \offset)/1.414},{(\eps - \i*\kappa)/1.414}) {};}{}
            \ifthenelse{\i=1}{\draw (x) -- (b);}{}
            \ifthenelse{\i=\m-1}{\draw (x) -- (b);}{}
            \ifthenelse{\i=\middle}{\draw (b) -- (1-\eps, 1);}{}
            \ifthenelse{\i=\middlee}{\draw (b) -- (1, 1-\eps);}{}
        }
        \draw (1-\eps,1) -- (1,1-\eps);
    \end{tikzpicture}
    \caption{The construction used in \Cref{thm:worstcase}. Left: a noncrossing tour of length $\Omega(n)$.
    Right: a tour of length $O(1)$.}
    \label{fig:worstcase}
\end{figure}

\section{Worst Case}\label{sec:worstcase}

We construct a worst-case instance in which there exists a
noncrossing tour with length $\Omega(n)$, as well as a tour
of constant length. The construction we use is depicted in \Cref{fig:worstcase}.

\begin{theorem}\label{thm:worstcase}
    Let $n \in \mathbb{N}$ be even.
    For $\epsilon > 0$ sufficiently small,
    there exists an instance of the Euclidean TSP in the plane
    where X-opt has approximation ratio at least $\frac{n}{2} \cdot (1 - \epsilon)$.
    In particular, the approximation ratio can be brought arbitrarily close to $\frac{n}{2}$.
\end{theorem}

\begin{proof}
    We place one point $s$ at $(0, 0)$. Next, we place
    $k = n/2$ points equally spaced along the line segment extending from
    $(1 - \epsilon/2, 1)$ to $(1, 1-\epsilon/2)$. We label these points $\{y_i\}_{i=1}^k$,
    ordering them by increasing $x$-coordinate.
    
    Consider the cone $K$ with vertex at the origin, defined by all conic combinations of
    $\{y_1, y_k\}$. Define the height along the axis of $K$ of a point $a$ by 
    the distance of $a$ from the origin along the axis of $K$.
    We place $k-1$ points
    along the line segment perpendicular to the axis of $K$ at a height
    $\epsilon/\sqrt{8}$, excluding its endpoints. We label these
    points $\{x_i\}_{i=1}^{k-1}$, sorting them again by increasing $x$-coordinate.
    Note that it does not matter where exactly we place these points, as long
    as no two points are placed in the same location.
    Observe that we have now placed exactly
    $2k = n$ points inside $[0, 1]^2$.
    
    To draw a noncrossing tour, we start at $s$, and draw the edge
    $\{s, y_1\}$. We then draw the edges $\{y_i, x_i\}_{i=1}^{k-1}$.
    Lastly, we add the edges $\{x_{k-1}, y_k\}$ and $\{y_k, s\}$, which closes up
    the tour.
    
    By construction, this tour contains no intersecting edges. To bound
    its length from below, observe that all edges have a length of at least
    $\sqrt{2} - \epsilon/\sqrt{2} = \sqrt{2}(1 - \epsilon/2)$. Thus, this tour has a length
    of at least $\sqrt{2}n(1-\epsilon/2)$.
    
    We now bound the length of the optimal tour from above, by
    \[
        d(s, x_1) + d(x_1, y_1) + d(s, x_{k-1}) + d(x_{k-1}, y_k)
            + d(x_1, x_k) + d(y_1, y_k)
                \leq 2\sqrt{2}(1+\epsilon/2).
    \]
    Putting these bounds together, we find a ratio of
    \[
        \frac{\sqrt{2}n(1-\epsilon/2)}{2\sqrt{2}(1+ \epsilon/2)} 
            = \frac{n}{2} \cdot \frac{1 - \epsilon/2}{1 + \epsilon/2}
            \geq \frac{n}{2} \cdot (1-\epsilon),
    \]
    as claimed. \qed
\end{proof}

\begin{remarknonum}
    The construction used for \Cref{thm:worstcase} only holds for $n$ even. For
    odd $n$, we can use a similar construction, but the approximation ratio then becomes
    $(n-1)/2 \cdot (1-\epsilon)$.
\end{remarknonum}

A simple argument shows that the approximation ratio given in \Cref{thm:worstcase}
is essentially as bad as one can get in the metric TSP. Given any instance,
let $x$ and $y$ be those points separated by the greatest distance. Any tour must
travel from $x$ to $y$ and  back to $x$ again, so any tour is of length at least
$2\cdot d(x, y)$. Moreover, every tour contains exactly $n$ edges, so any tour has length
at most $n\cdot d(x, y)$. Hence, the approximation ratio of any algorithm for the metric 
TSP is at most $n/2$.

\section{Average Case}

Although the worst-case construction of \Cref{sec:worstcase}
shows that the uncrossing heuristic may yield
almost as bad of an approximation as is possible for TSP, it is possible
that the heuristic still shows good behavior on average. To exclude this
possibility, we consider a standard average-case model wherein $n$ points
are placed uniformly and independently in the plane. We then
construct a tour of length $\Omega(n)$ in expectation.
We present our results in \Cref{thm:averagecase}.

To simplify our arguments, it would be convenient
if we could consider only a subset of all points, so that
we can look at a linear-size sub-instance with nicer properties.
Constructing a long noncrossing tour through this subset would be much
easier.

Since an optimal tour through a set of points $X$ is always
at least as long as an optimal tour through a subset
$Y \subset X$, it is tempting to conjecture that
something similar holds for all noncrossing tours. Perhaps,
given a noncrossing tour through $Y \subseteq X$, we can
extend the tour to all of $X$ without decreasing its length?

Unfortunately, this turns out to be false. We provide a counterexample in
\Cref{fig:counterexample}.

\begin{figure}
    \centering
    \begin{tikzpicture}[scale=2]

        \vertex[label=left:$x$] (x) at (0, 0) {};
        \vertex[label=left:$a$] (a) at (90: 0.5) {};
        \vertex[label=left:$b$] (b) at (210: 0.5) {};
        \vertex[label=right:$c$] (c) at (-30: 0.5) {};

        \vertex[label=right:$u$] (u) at (-18: 1) {};
        \vertex[label=left:$v$] (v) at (-18 + 120: 1) {};
        \vertex[label=left:$w$] (w) at (-18 + 240: 1) {};

        \draw[
            decoration={
        brace,
        mirror,
        raise=0.2cm
    }, decorate] (b) -- (c) node[sloped, pos=0.5, below=0.3cm] {$1$};
    \draw[dashed] (b) -- (c);
    \draw[dashed] (v) -- (c);
    
        \draw[
            decoration={
        brace,
        mirror,
        raise=0.2cm
    }, decorate] (c) -- (v) node[sloped, pos=0.5, above=0.3cm] {$L$};
             
    \end{tikzpicture}\hspace{2em}
    \begin{tikzpicture}[scale=2]

        \vertex[] (x) at (0, 0) {};
        \vertex[] (a) at (90: 0.5) {};
        \vertex[] (b) at (210: 0.5) {};
        \vertex[] (c) at (-30: 0.5) {};

        \vertex[] (u) at (-18: 1) {};
        \vertex[] (v) at (-18 + 120: 1) {};
        \vertex[] (w) at (-18 + 240: 1) {};

        \draw (b) -- (u) -- (c) -- (v) -- (a) -- (w) -- (b);

    \end{tikzpicture}\hspace{2em}
     \begin{tikzpicture}[scale=2]

        \vertex[] (x) at (0, 0) {};
        \vertex[] (a) at (90: 0.5) {};
        \vertex[] (b) at (210: 0.5) {};
        \vertex[] (c) at (-30: 0.5) {};

        \vertex[] (u) at (-18: 1) {};
        \vertex[] (v) at (-18 + 120: 1) {};
        \vertex[] (w) at (-18 + 240: 1) {};

        \draw (x) -- (w);
        \draw (c) -- (w); \draw (a) -- (u); \draw (b) -- (v);
        \draw (c) -- (u); \draw (a) -- (v);
        \draw (b) -- (x);

    \end{tikzpicture}
    \caption{Left: instance described in \Cref{lemma:counterexample}. 
    Center: the tour $T$ from \Cref{lemma:counterexample}.
    Right: the longest possible noncrossing tour through all points
    in the instance.}
    \label{fig:counterexample}
\end{figure}
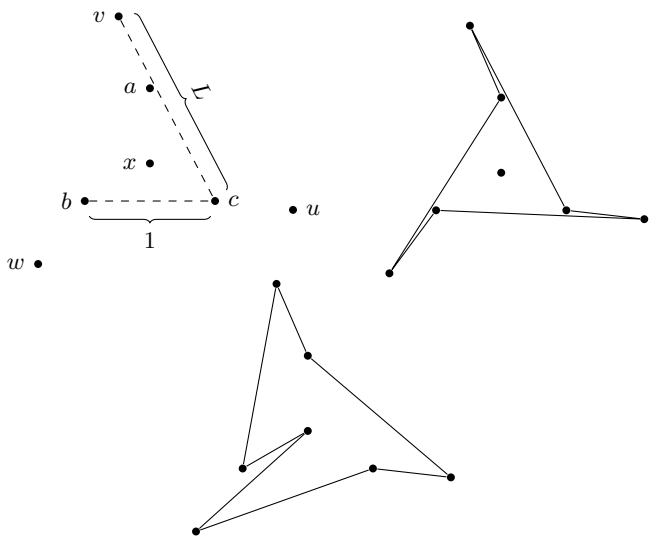

\begin{theorem}\label{lemma:counterexample}
    There exists a set of points in the plane, together with a noncrossing 
    tour $T$ through all but one of these points, such that all noncrossing tours
    through all points have length strictly less than $\ell(T)$.
\end{theorem}

\begin{proof}
    Consider the instance shown in \Cref{fig:counterexample}, along
    with the tour $T = avcubwa$. Note that $T$ passes through all points,
    except for the central point $x$. We construct this instance such that
    $d(b, c) = d(a, b) = d(a, c) = 1$, so that $a$, $b$ and $c$ form the
    vertices of an equilateral triangle. We also set the distances
    $d(v, c) = d(b, u) = d(a, w) = L$, where $L \gg 1$. We set the angle
    $\angle cub = \angle avc = \angle a w b$ small, so that
    $d(c, u) = d(v, a) = d(b, w) = L - 1 + o(1)$; here,
    $o(1)$ means terms decreasing in $L$. For instance,
    setting this angle to $1/L$ would suffice.
    Then we have $d(w, c) = L - \frac{1}{2} + o(1)$.
    By construction, $d(x, a) = d(x, b) = d(x, c) = 1/\sqrt{3}$.
    We can furthermore compute $d(u, v) = d(v, w) = d(u, w) = \sqrt{3} L + o(1)$,
    $d(a, u) = d(b, v) = d(c, w) = L - \frac{1}{2} + o(1)$, and
    $d(x, u) = d(x, v) = d(x, w) = L - \frac{1}{4\sqrt{3}} + o(1)$.
    Since we will compare sums of these various distances, all inequalities
    in the remainder of the proof are understood to hold for sufficiently large $L$.

    We classify the edges into types:

    \begin{table}[h!]
    \centering
    \begin{tabular}{l|l|l|l|l|l|l|l}
         Type & 1 & 2 & 3 & 4 & 5 & 6 & 7 \\ \hline
         Length (up to $o(1)$ terms) & $\frac{1}{\sqrt{3}}$ & 1 & $L-1$ & $L-\frac{1}{2}$ & $L-\frac{1}{4\sqrt{3}}$ & $L$ & $\sqrt{3}L$
    \end{tabular}
    \end{table}
    We furthermore call the edges of types 1 and 2 the short edges,
    of types 3 through 6 the long edges, and of type 7 the very long
    edges.
    Finally, we classify the points $\{a, b, c, u, v, w\}$ into the near points
    $N = \{a, b, c\}$ and the far points $F = \{u , v, w\}$.

    We start by computing the length of $T$:
    \[
        \ell(T) = 3L + 3(L-1) + o(1) = 6L - 3 + o(1).
    \]
    We now show that all noncrossing tours that include $x$ are shorter than
    $6L - 3 + o(1)$.

    As it is laborious to check all possible noncrossing tours, we begin by excluding some
    possibilities. First, suppose the tour contains $k$ very long edges. Then the tour
    can contain at most $6-2k$ long edges, since each long edge has exactly one endpoint
    in $F$, and $k+1$ of these have been used. Such a tour
    has a length of at most
    $k\sqrt{3} + (6-2k)L + O(1) < 6L - 3 + o(1)$. Thus, any sufficiently long
    tour cannot contain any very long edges.

    Now suppose that a tour contains two short edges. Since very long edges are excluded,
    such a tour has length at most $5L + 2 + o(1) < 6L - 3 + o(1)$.
    Therefore, the tour must consist of
    six long edges and one short edge.

    Suppose the unique short edge in our tour is of type 2. Without loss
    of generality, we assume it is $\{a, b\}$. Since all other
    edges are long, both $a$ and $b$ must connect to points in $F$. Suppose
    $a$ connects to $u$. Then $x$ must connect to $u$ and $w$, since
    otherwise we would introduce a second short edge. But this makes it
    impossible to connect $c$ to $v$ without creating an intersection or a subtour.
    Thus, $a$ cannot connect to $u$.
    
    By identical reasoning, we cannot connect $b$ to $u$. The only option
    is thus to connect $a$ to $v$ and $b$ to $w$. As very long edges are excluded, $v$
    must connect to $c$ (connecting to $x$ would yield an intersection with $L(a, b)$) and
    $x$ must connect to $w$. But this makes it impossible
    to connect $u$ to the tour without creating an intersection. 
    This shows that the short edge can only be of type 1.

    Suppose the tour contains an edge of type 6. Without loss of generality, we assume this to be
    $\{b, u\}$. Then observe that $w$ can only be connected to the rest
    of the instance by connecting it to $b$ or $c$, since very long edges are excluded.
    Hence, connecting $w$ would yield a subtour, and so no edges of type 6 are possible.

    We now note that the tour cannot contain all three edges of type 5, since these all connect
    to $x$.
    Suppose the tour contains
    two such edges, say, $\{x, v\}$ and $\{x, w\}$. Then it is not possible to
    connect $b$ to the rest of the instance without creating an intersection or a subtour.
    Hence, at most one edge of type 5 is permissible.

    The longest tour now contains
    one edge of type 5, three edges of type 4, two edges of type 3 and one edge of type 1,
    which yields a length of
    \[
        L - \frac{1}{4\sqrt{3}} + 3\left(L-\frac{1}{2}\right) + 2(L-1) + \frac{1}{\sqrt{3}} + o(1)
            = 6L - 3.5 + \frac{3}{4\sqrt{3}} + o(1) < 6L - 3 + o(1).
    \]
    Therefore, all noncrossing tours through all points of the instance have
    length strictly less than $\ell(T)$. \qed
\end{proof}

\Cref{lemma:counterexample} shows 
that in constructing a tour through a random instance, we must carefully make sure
to take all points into account. Any points we leave behind could
reduce the length of the tour if we attempt to add them
after constructing a long subtour.

As the worst-case tour consists of many long almost-parallel edges, we seek to construct
a similar tour in random instances.
Our strategy will consist of dividing part of the unit square into many long parallel strips,
and forming noncrossing Hamiltonian paths within these strips. We then connect paths
of adjacent strips without creating intersections, forming a long Hamiltonian path through all strips
together. The endpoints can then be connected if we leave out some space for points
along which to form a connecting path. See \Cref{fig:averagecase} for a schematic
depiction of our construction.

Before we proceed to the proof of \Cref{thm:averagecase}, we need some simple lemmas.
We start with a lemma that bounds the probability that any region in
$[0,1]^2$ contains too few points for our construction to work.

\begin{lemma}\label{lemma:nonempty}
    Let $A \subseteq [0,1]^2$, and let $X$ be a finite set of $n$ points
    placed independently uniformly at random in $[0, 1]^2$. Let
    $N = |A \cap X|$. Then
    \[
    \prob(N \leq k) \leq e^k e^{-\frac{k^2}{2n\mathrm{area}(A)}}
        e^{-n\cdot \mathrm{area}(A)/2}
    \]
    for $k \leq n \cdot \mathrm{area}(A)$.
\end{lemma}

\begin{proof}
    For $x \in X$, let $S(x)$ be an indicator variable taking a value of
    $1$ iff $x \in A$, and 0 otherwise. Then $N = \sum_{x \in X} S(x)$.
    Let $\mu = \expect(N) = n\cdot \mathrm{area}(A)$. By Chernoff's bound,
    for $\delta \in (0, 1)$,
    \[
        \prob(N \leq (1-\delta)\mu)
            \leq e^{-\delta^2 \cdot n \cdot \mathrm{area}(A)/2}.
    \]
    The result now follows from setting
    $\delta = 1 - k/\mu$, which implies $\delta^2 = 1 + k^2/\mu^2 - 2k/\mu$,
    and inserting this into the above bound. \qed
\end{proof}

The following observation and lemma are required to form suitable Hamiltonian paths
through subsets of our random instance.

\begin{observation}\label{lemma:intersection}
    Let $a, b, c, d$ be four distinct points in the plane, no three of which are
    collinear. Suppose $L(a, b)$ intersects $L(c, d)$. Then
    $L(a, d)$ cannot intersect $L(b, c)$. Moreover, $d(a, b) + d(c, d) > d(a, d) + d(b, c)$.
\end{observation}

\begin{lemma}\label{lemma:hamiltonian_path}
    Let $X$ be a set of distinct points in the plane, and
    assume no three points in $X$ are collinear. Fix distinct points $s, t \in X$.
    Then there exists a noncrossing Hamiltonian path through $X$ with endpoints
    $s$ and $t$.
\end{lemma}

\begin{proof}
    Fix an arbitrary, possibly crossing Hamiltonian path with endpoints $s$ and $t$.
    Suppose the edges $e = \{a, b\}$ and $f = \{c, d\}$ intersect. By 
    Observation \ref{lemma:intersection},
    $e' = \{a, d\}$ and $f' = \{b, c\}$ do not intersect. We assume that the path
    remains connected if we exchange $e$ and $f$ for $e'$ and $f'$ (if this fails, then
    we swap $a$ for $b$).

    Observation \ref{lemma:intersection} also shows that the length of the resulting path is strictly
    smaller than the length of the original path. We repeat this process, removing intersections
    until we obtain a noncrossing path. This process must terminate in a finite number of steps,
    since the number of Hamiltonian paths is finite and each step strictly decreases the path length.
    Hence, no path is seen twice.

    It remains to show that $s$ and $t$ remain the endpoints throughout this process.
    Observe that a point is an endpoint of the path if and only if
    it has degree 1. Since the exchange operation preserves the degree of all vertices
    in the path, the endpoints of the path do not change. \qed
\end{proof}

The next lemma is useful to connect Hamiltonian paths in neighboring rectangular
regions. 

\begin{lemma}\label{lemma:path_extension}
    Let $A$, $B$ and $C$ be distinct rectangular regions in the plane. Assume
    $B$ shares an edge with $A$ and with $C$, but $A$ and $C$ are disjoint except possibly
    in a single point. Let $S_A$, $S_B$ and $S_C$ be finite
    sets of points in $A$, $B$ and $C$ respectively.
    Let $x_A$ and $x_C$ be the points in $B$ closest to $A$ and $C$,
    and assume $x_A \neq x_C$.
    Let $P$ be a noncrossing Hamiltonian path through $S_B$ with endpoints $x_A$ and $x_C$.
    Let $P'$ be a path obtained by connecting $x_A$ to any point in $S_A$ and $x_C$
    to any point in $S_C$. Then $P'$ is noncrossing.
\end{lemma}

\begin{proof}
    Without loss of generality, we assume the regions $A$, $B$ and $C$ are aligned at the horizontal axis.
    Moreover, we assume (again without loss of generality) that $C$ borders $B$ to the right. This leaves
    three cases to examine for $A$: it may border $B$ to the left, to the bottom, or to the top. The latter
    two cases are identical, so we examine only the first two.

    Suppose $A$ borders $B$ to the left. Let $y_A  \in S_A$, and suppose we extend $P$ by
    connecting $y_A$ to $x_A$. Because $P$ is noncrossing, we need only check whether the edge
    $e = \{y_A, x_A\}$ intersects any edge of $P$. Suppose such an edge, say $e^*$, exists. Then this implies
    that one endpoint of $e^*$ must lie to the left of $x_A$, which is a contradiction.

    Now consider the border of $B$ with $C$, and extend the path by adding the edge $f = \{y_C, x_C\}$ for some
    $y_C \in S_C$. By the same reasoning, $f$ cannot intersect any edge of $P$. It remains to 
    check whether $e$ can intersect $f$. Observe that $L(e)$ lies entirely to the left of $x_A$,
    which lies to the left of $x_C$, which in turn lies to the left of the entirety of $L(f)$.
    Thus, $e$ and $f$ cannot intersect.

    Now suppose $A$ borders $B$ to the top. The argument from the previous case now fails,
    since we can no longer order the extending edges from left to right.
    Suppose that $e$ intersects $f$ in the point $q \in B$. We add a direction to $e$ and $f$, taking
    $x_A$ and $x_C$ as their origin. Since $x_C$ lies below $x_A$, we know that $e$ must lie
    below $f$ after passing through $q$, and stay above $f$ until it reaches its endpoint.
    Since the endpoints of $e$ and $f$ lie outside of $B$, this implies that either the point where
    $e$ exits $B$ lies below the point where $f$ exits $B$, or both edges exit
    $B$ in the same border of $B$. This is a contradiction, and so we are done. \qed
\end{proof}

The proof of \Cref{lemma:path_extension}
fails when the points $x_A$ and $x_C$ are identical, or equivalently,
when the point set contained in $B$ is not nice for $B$.
The following lemma shows that, provided $B$ contains enough
randomly placed points, this is unlikely to occur.

\begin{lemma}\label{lemma:nice}
    Let $A$ be a rectangular region in the plane. Let $X$ be a set of $n \geq 2$
    points placed uniformly at random in $A$. 
    The probability that $X$ is not nice for $A$ is
    at most $6/n$.
\end{lemma}

\begin{proof}
    Let $e_i$, $i \in [4]$, be any edge of $A$. 
    Let $E_{ij}$ denote the event that $x_{e_i} = x_{e_j}$ for $i \neq j$.
    Then we seek to bound $\prob(\cup_{i=1}^4 \cup_{j = i+1}^4 E_{ij})$
    from above.
    
    Without loss of generality,
    we assume that $e_i$ lies along the $x$-axis, and that $A \setminus e_i$
    lies above $e_i$, i.e., $e_i$ is the bottom edge of $A$.

    Let $e_j$ be a different edge of $A$. If $e_j$ is the top edge
    of $A$, then the event $E_{ij}$ occurs with probability 0,
    since it requires all points to lie on the same horizontal line.

    Suppose now that $e_j$ is the right edge of $A$, and let
    $y \in X \setminus \{x_{e_i}\}$. Observe that the horizontal coordinates
    of all points in $X$ are independent uniform random variables.
    Thus, the probability that $x_{e_i}$ is the point with
    the largest horizontal coordinate is $1/n$ by symmetry.
    Therefore,
    $\prob(E_{ij}) = 1/n$.

    To conclude, we apply a union bound to obtain
    \[
        \prob\left(\cup_{1 \leq i < j \leq 4} E_{ij}\right)
            \leq \sum_{1 \leq i < j \leq 4} \prob\left(E_{ij}\right)
            = \frac{6}{n}
    \]
    as claimed. \qed
\end{proof}

The final lemma we require follows from elementary calculus and probability theory.

\begin{lemma}\label{lemma:uniform_expect}
    Let $X$ and $Y$ be uniformly distributed over $[a, b]$. Then
    \[
        \expect(|X - Y|) = \frac{b-a}{3}.
    \]
\end{lemma}

\begin{proof}
    As $X$ and $Y$ are independent, we can compute the required quantity directly
    using their joint distribution. Let $f_X$ and $f_Y$ denote their respective
    probability density functions. We have
    \begin{align*}
        \expect(|X-Y|) &= \int_a^b\int_a^b |x - y| f_X(x)f_Y(y) \dd x \dd y\\
            &= \frac{1}{(b-a)^2}\int_a^b\int_a^b |x - y| \dd x \dd y.
    \end{align*}
    Substituting $x = a + (b-a)\bar x$ and $y = a + (b-a)\bar y$, the
    integral reduces to $(b-a)^3\int_0^1\int_0^1 |\bar x - \bar y|\dd \bar x \dd \bar y
    = (b-a)^3/3$, completing the proof. \qed
\end{proof}

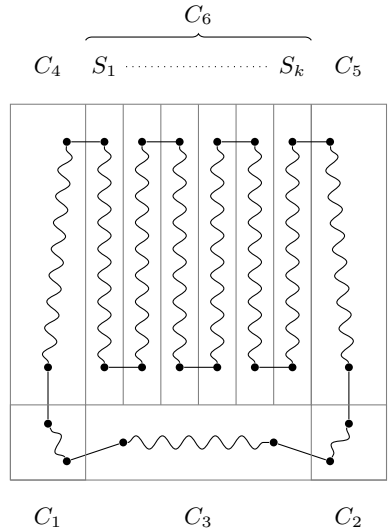
\begin{figure}
    \centering
    \begin{tikzpicture}[scale=5]
        \pgfmathsetmacro\k{6};

        \pgfmathsetmacro\c{0.2};
        
        \node[] at (0.5*\c, 1.1) {$C_4$};
        \node[] at (1.0-0.5*\c, 1.1) {$C_5$};
        \draw[
            decoration={
        brace,
    }, decorate] ((\c,1.17) -- (1-\c,1.17) node[sloped, pos=0.5, above=0.1cm] {$C_6$};

        \node [] at (0.5*\c,-0.1) {$C_1$};
        \node [] at (1-0.5*\c,-0.1) {$C_2$};
        \node [] at (0.5,-0.1) {$C_3$};
        
        \vertex[] (xr) at (0.75*\c, 0.25*\c) {};
        \vertex[] (xt) at (0.5*\c, 0.75*\c) {};
        \vertex[] (yt) at (1-0.5*\c, 0.75*\c) {};
        \vertex[] (yl) at (1-0.75*\c, 0.25*\c) {};

        \vertex[] (z) at (0.5*\c, 1.5*\c) {};
        \vertex[] (w) at (0.75*\c, 1-.5*\c) {};
        \vertex[] (p) at (1.5*\c, 0.5*\c) {};
        \draw[snake it] (xr) -- (xt);
        \draw (xr) -- (p);
        \draw (xt) -- (z);
        \draw[snake it] (z) -- (w);
        
        \vertex[] (s) at (1-0.5*\c, 1.5*\c) {};
        \vertex[] (t) at (1-0.75*\c, 1-.5*\c) {};
        \vertex[] (q) at (1-1.5*\c, 0.5*\c) {};
        \draw[snake it] (yt) -- (yl);
        \draw (yt) -- (s);
        \draw (yl) -- (q);
        \draw[snake it] (s) -- (t);
        
        \draw[snake it] (p) -- (q);
        
        \draw[gray] (0,0) -- (1,0) -- (1,1) -- (0,1) -- (0,0);
        
        \draw[gray] (0, 0) rectangle ++ (\c,\c);
        \draw[gray] (1-\c, 0) rectangle (1,\c);
        \draw[gray] (\c, \c) -- (1-\c, \c);
        \draw[gray] (\c, \c) -- (\c, 1);
        \draw[gray] (1-\c, \c) -- (1-\c, 1);

        \pgfmathsetmacro\m{\k-1};
        \pgfmathsetmacro\n{\k+1};
        \foreach \i in {1,...,\k} {

            \pgfmathsetmacro\x{\c+\i*(1-2*\c)/(\k)};
            \pgfmathsetmacro\y{\c+(\i+1)*(1-2*\c)/(\k)};
            \ifthenelse{\i<\n}{
                \draw[gray] (\x,1) -- (\x,\c);
            }{};
            
            \vertex[] (a) at ({\x-0.5*(1-2*\c)/\k},1.5*\c) {}; 
            \vertex[] (b) at ({\x-0.5*(1-2*\c)/\k},1-0.5*\c) {}; 

            \draw[snake it] (a) -- (b);
            \ifodd\i{\draw (a) -- ({\y-0.5*(1-2*\c)/\k},1.5*\c);}\else
            {\ifthenelse{\i<\n}{\draw (b) -- ({\y-0.5*(1-2*\c)/\k},1-0.5*\c);}{}}\fi
            
            \ifthenelse{\i=1}{\draw(w) -- (b); \node[] (s1) at (0.5*\x+0.5*\c,1.1) {$S_1$};}{};
            \ifthenelse{\i=\n}{\draw(t) -- (b);}{};
            
            \pgfmathsetmacro\textx{0.5*\x + 0.5*(1-\c)};
            \ifthenelse{\i=\m}{\node[] (sk) at (\textx,1.1) {$S_k$};}{};
        }

    \draw[dotted] (s1) -- (sk);
      
    \end{tikzpicture}
    \caption{The construction used in \Cref{thm:averagecase}. Wavy lines
    represent Hamiltonian paths within the rectangular region they lie in, while
    straight lines represent edges.}
    \label{fig:averagecase}
\end{figure}

We are now in a position to prove \Cref{thm:averagecase}. See \Cref{fig:averagecase} for
a sketch of the construction we use in the proof.

\begin{theorem}\label{thm:averagecase}
    Suppose a TSP instance is formed by placing $n$ points uniformly at random in the 
    unit square.
    Then the expected value of the ratio of the worst local optimum of X-opt
    and the optimal tour on this instance is $\Omega(\sqrt{n})$.
\end{theorem}

\begin{proof}
    We begin by partitioning the unit square into six rectangular regions.
    Let $c \in (0, 1)$ be a constant, to be fixed later.
    Let $C_1$ denote the square region with opposite points
    $(0, 0)$ and $(c, c)$, and let $C_2$ similarly denote the square region with
    corner points $(1-c, 0)$ and $(1, c)$.
    Next, let $C_3$ be the rectangular region
    with opposite corners $(c, 0)$ and $(1-c, c)$.
    The region $C_4$ denotes the rectangular region
    with corner points $(0, c)$ and $(c, 1)$, while 
    $C_5$ denotes the region
    with corners $(1-c, c)$ and $(1, 1)$.
    Finally, the region $C_6$ denotes the rectangle with
    corners $(c, c)$ and $(1-c, 1)$.

    We divide $C_6$ into vertical strips of width $\alpha \cdot \frac{1-2c}{n}$
    for some $\alpha > 0$ to be fixed later, so that there are
    $k = \lfloor n/\alpha \rfloor$ strips in total. We label the strips
    from left to right as $\{S_i\}_{i=1}^k$.

    Let $X$ be a set of $n$ points placed uniformly at random in $[0, 1]^2$.
    Note that, by \Cref{lemma:nonempty}, the probability that
    any of $X_{C_i}  := C_i \cap X$ for $i \in [5]$ contains fewer than 31 points
    is at most $5e^{-\Omega(n)}$. Hence, we assume for the remainder
    of the proof that these regions contain at least 31 points.
    The possibility that this is not true reduces the expected
    tour length by a factor of at most $1-e^{-\Omega(n)}$,
    which does not affect the result.

    Observe that
    each $x \in X_{C_i}$ is uniformly distributed over $C_i$. Since we assume
    that each $X_{C_i}$ contains at least 31 points, we see by \Cref{lemma:nice} that the probability that
    $X_{C_i}$ is not nice for $X$ is at most $6 / 31$.
    By a union bound, the probability that $X_{C_i}$ is not nice
    for $C_i$ for any $i \in [5]$ is at most
    $5\cdot \frac{6}{31} = \frac{30}{31}$. Hence, we assume for the remainder
    that each $X_{C_i}$ is nice for $C_i$, $i \in [5]$. The possibility that this is not
    true reduces the expected tour length by at most a factor of $\frac{1}{31}$, which
    does not affect the result.

    Consider strip $S_i$. Let $l_i$ be the point in $S_i \cap X$ with
    the smallest $x$-coordinate, the \emph{leftmost} point, provided it exists. Similarly,
    $r_i$ denotes the point in $S_i \cap X$ with the largest $x$-coordinate,
    or the \emph{rightmost} point. Since the probability
    that any three points in $X$ are collinear is 0,
    we can use \Cref{lemma:hamiltonian_path} to establish the existence of
    a noncrossing Hamiltonian path through $S_i$ with endpoints $l_i$ and $r_i$. Let
    $P_{S_i}$ be such a path.

    After forming the paths $P_{S_i}$ for all $i \in [k]$, we connect 
    $P_{S_i}$ to $P_{S_{i+1}}$ for $i \in [k-1]$. To do this, we
    simply connect
    $r_i$ to $l_{i+1}$. The result is a Hamiltonian path $P_6$ through $C_6$,
    which contains the paths $P_i$ as subpaths.
    By \Cref{lemma:path_extension}, we know that $P_6$ is noncrossing.

    Next, we form noncrossing Hamiltonian paths through the remaining regions.
    For the region $C_4$, we let the endpoints of the path be the rightmost and bottom-most
    points in $X \cap C_4$. For $C_5$, the endpoints are the leftmost and bottom-most points.
    For $C_1$, we take the top-most and rightmost points, for $C_2$ the leftmost and top-most,
    and for $C_3$ the leftmost and rightmost. We label the path through region $C_i$
    by $P_i$, $i \in [6]$. Observe that these endpoints are distinct in all cases, since
    we assume that $X_{C_i}$ is nice for $C_i$ for each $i \in [5]$.
    
    We now connect these paths as follows.
    Let $C_i$ be any of the regions. If the region shares a border with $C_j$, excluding the bottom
    border of $C_6$, then
    we connect the points from $C_i$ and $C_j$ closest to the border. Observe that this is exactly the process
    we used to form $P_6$. Again using \Cref{lemma:hamiltonian_path}, the tour $T$ so formed is noncrossing.

    To bound the length of $T$ from below, we consider the length of the paths $P_{S_i}$, $i \in [k]$.
    The sum of these lengths is clearly a lower bound for $\ell(T)$. We thus have
    \[
        \expect(\ell(T)) \geq \sum_{i=1}^k \expect(\ell(P_{S_i})),
    \]
    by linearity of expectation. Moreover, observe that
    $\ell(P_{S_i}) = 0$ if fewer than 2 points are placed in $S_i$. Thus, we find by
    the law of total expectation
    \[
        \expect(\ell(P_{S_i})) = \expect(\ell(P_{S_i})\given |S_i \cap X| \geq 2)
            \cdot \prob(|S_i \cap X| \geq 2).
    \]
    Assume the event $|S_i \cap X| \geq 2$ occurs. Let $x, y$ be any two points in $S_i \cap X$.
    Then by the triangle inequality, $\ell(P_i) \geq d(x, y) \geq d_v(x, y)$, where $d_v$ denotes the vertical distance
    between $x$ and $y$. Observe that the vertical coordinates of $x$ and $y$ are independent
    and uniformly distributed over $[c, 1]$. This implies by \Cref{lemma:uniform_expect} that
    \[
        \expect(\ell(P_i) \given |S_i \cap X| \geq 2)
            \geq \frac{1-c}{3}.
    \]
    Using \Cref{lemma:nonempty} to bound $\prob(|S_i \cap X| \geq 2)$ from below, we have
    \[
        \expect(\ell(T)) \geq \left\lfloor \frac{n}{\alpha}\right\rfloor \cdot 
            \frac{1-c}{3}
            \cdot \left(1 - e^{1-\frac{1}{2}\alpha(1-2c)(1-c)} \cdot e^{-\frac{1}{2\alpha(1-2c)(1-c)}}\right),
    \]
    where we use the fact that $\mathrm{area}(S_i) = \alpha(1-c)(1-2c)/n$.
    
    It remains to fix values for $c$ and $\alpha$ such  that this expectation is nontrivial;
    for instance, $\alpha = 10$ and $c = 0.1$ suffice. We then find
    $\expect(\ell(T)) = \Omega(n)$. Since there exists a tour of length $O(\sqrt{n})$ in the Euclidean
    TSP with high probability \cite{friezeProbabilisticAnalysisTSP2007}, we are done. \qed
\end{proof}

\section{Practical Performance of Uncrossing Tours}\label{sec:numerical}

In this section, we show that in practical instances there is a large gap
with the results suggested by \Cref{thm:averagecase}. We generate
instances with $n$ points sampled from the uniform distribution
over $[0, 1]^2$,
and run X-opt on these instances. As a starting tour, we
pick a tour from the uniform distribution on all tours.
We compute the lengths of the locally optimal tours obtained
from our implementation of X-opt,
and average them for each fixed value of $n$ we evaluate.
We consider the simplest possible pivot rule: starting from
an arbitrary edge $e$ in the tour, we check whether $e$ intersects
with any other edge, performing an exchange when we find the first such edge.
If we do not find such an intersecting edge, we move on to the next edge
in the tour and repeat the process. By ``next'', we mean that we order
the edges of a tour according to the permutation on the vertices by which we represent
the tour.

Since the optimal tour length
is $\Theta(\sqrt{n})$ with high probability \cite{friezeProbabilisticAnalysisTSP2007},
we compare the length of the tours
we obtain
with this function. Their ratio
then serves as a proxy for the approximation ratio of X-opt.
We perform this procedure for
$n \in \{100, \ldots, 1000, 2000,\ldots, 10000\}$. For each value of $n$, we take
$N = 16,000$ samples. The results are
shown in \Cref{fig:numerical}. To the precision we are able to obtain, we
cannot distinguish the approximation ratio from constant.

\begin{figure}
    \centering
    \includegraphics[width=\textwidth]{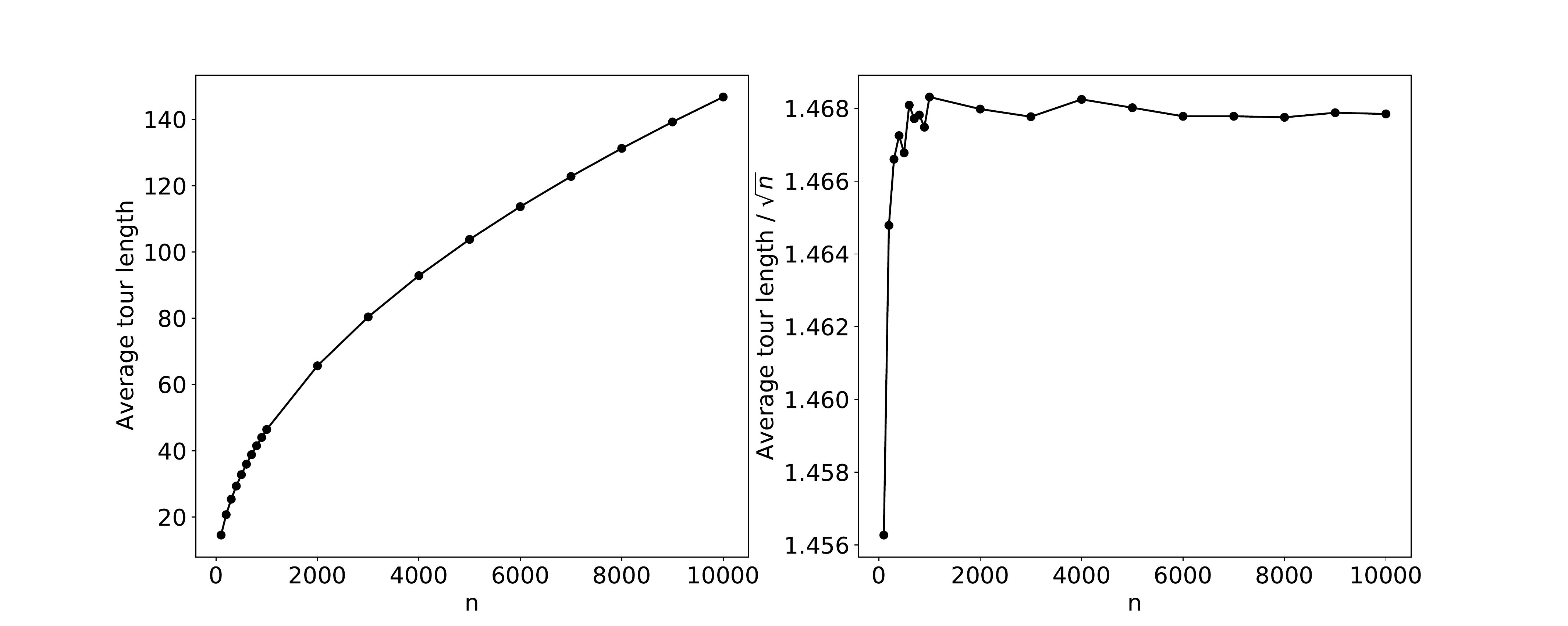}
    \caption{Numerical evaluation of the average-case performance of X-opt.}
    \label{fig:numerical}
\end{figure}

\section{Discussion}

Although the results we presented in \Cref{thm:worstcase}
and \Cref{thm:averagecase}
are rather negative for X-opt, the numerical experiments
of \Cref{sec:numerical} paint a much more optimistic picture.
The heuristic appears to be much more efficient in practice than our
lower bounds suggest. Indeed, while \Cref{thm:averagecase} suggests
an approximation ratio of $\Omega(\sqrt{n})$, the numerical experiments
in \Cref{sec:numerical} suggest a constant approximation ratio.

One possible explanation for this discrepancy is that
we compare the optimal solution on
any instance to local optima specifically constructed to be bad. This is
a rather standard approach, and it is not surprising that
it gives pessimistic results. However, the results in this
case are especially pessimistic, considering that we
can show a tight lower bound for the expected tour length
in the average case.

We consider this to be an indication that this approach
is incapable of explaining the practical approximation performance
of local search heuristics. In order to more closely model the
true behavior of heuristics, it seems one must analyze the landscape
of local optima, and the probability of reaching different local optima.
We stress that this discrepancy cannot be resolved by other
standard methods of probabilistic analysis. In particular, smoothed
analysis \cite{spielmanSmoothedAnalysisAlgorithms2004} cannot help, since the smoothed
approximation ratio of an algorithm is bounded from below by the average-case
approximation ratio.

\bibliographystyle{splncs04}
\bibliography{bibliography}

\end{document}